\documentclass{journalA4}
\usepackage{amssymb,amsmath,latexsym,graphicx,wrapfig,capt-of,fancyhdr,url,subfigure,xspace,color}

\newcommand{\etal}{{\emph{et al.}\xspace}}
\newcommand{\void}[1]{}

\graphicspath{{figures/}}

\title{Angle-Restricted Steiner Arborescences\\ for Flow Map Layout\thanks{A preliminary version of this paper will appear at the 22nd International Symposium on Algorithms and Computation (ISAAC 2011). B. Speckmann and K. Verbeek are supported by the Netherlands Organisation for Scientific Research (NWO) under project no.~639.022.707.}}

\author{Kevin Buchin \and Bettina Speckmann \and Kevin Verbeek}

\date{
Dep. of Mathematics and Computer Science, TU Eindhoven, The Netherlands.\\ {\tt k.a.buchin@tue.nl} \qquad {\tt speckman@win.tue.nl} \qquad {\tt k.a.b.verbeek@tue.nl}
}

\begin{document}

\maketitle

\begin{abstract}
We introduce a new variant of the geometric Steiner arborescence problem, motivated by the layout of flow maps. Flow maps show the movement of objects between places. They reduce visual clutter by bundling lines smoothly and avoiding self-intersections. To capture these properties, our \emph{angle-restricted Steiner arborescences}, or \emph{flux trees}, connect several targets to a source with a tree of minimal length whose arcs obey a certain restriction on the angle they form with the source.

We study the properties of optimal flux trees and show that they are planar and consist of logarithmic spirals and straight lines. Flux trees have the \emph{shallow-light property}. We show that computing optimal flux trees is NP-hard. Hence we consider a variant of flux trees which uses only logarithmic spirals. \emph{Spiral trees} approximate flux trees within a factor depending on the angle restriction. Computing optimal spiral trees remains NP-hard, but we present an efficient 2-approximation, which can be extended to avoid ``positive monotone'' obstacles.
\end{abstract}

\section{Introduction}\label{sec:introduction}

Flow maps are a method used by cartographers to visualize the movement of objects between places~\cite{Dent1999,Slocum2010}. One or more sources are connected to several targets by arcs whose thickness corresponds to the amount of flow between a source and a target. Good flow maps share some common properties. They reduce visual clutter by merging (bundling) lines as smoothly and frequently as possible. Furthermore, they strive to avoid crossings between lines. \emph{Flow trees}, that is, single-source flows, are drawn entirely without crossings. Flow maps that depict trade often route edges along actual shipping routes. In addition, flow maps try to avoid covering important map features with flows to aid recognizability. Most flow maps are still drawn by hand and none of the existing algorithms (that use edge bundling), can guarantee to produce crossing-free flows.

\begin{figure}[t]
\centering
\includegraphics[height=4.5cm]{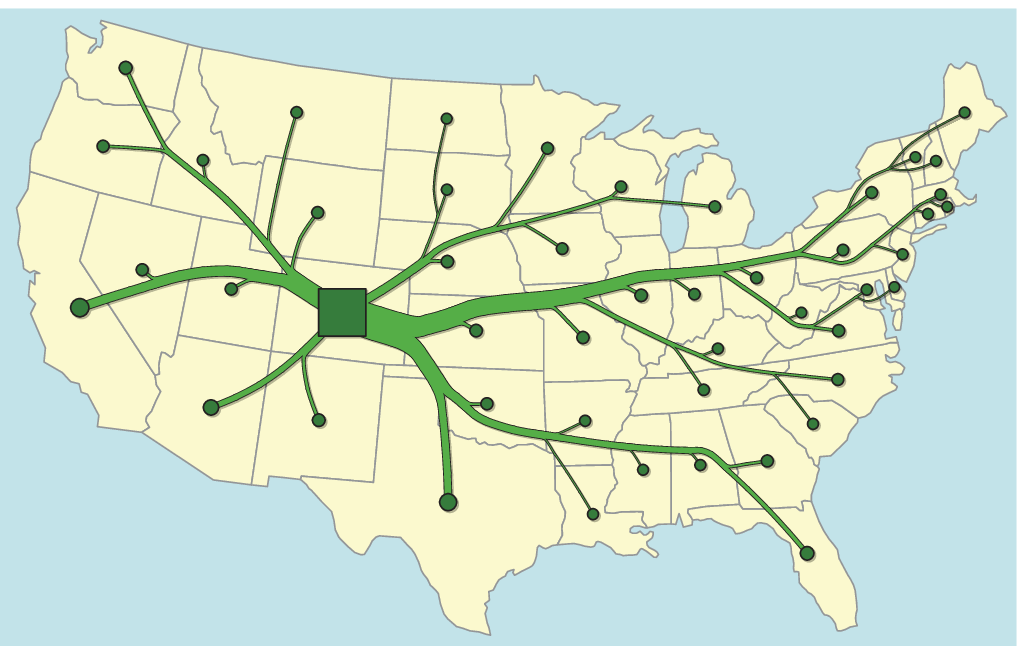}
\hfill
\includegraphics[height=4.5cm]{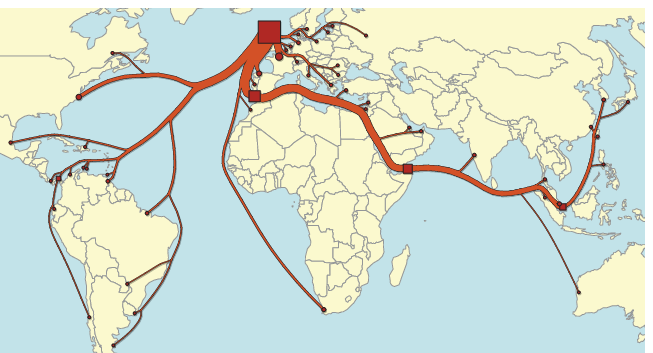}
\caption{Flow maps from our companion paper~\cite{InfoVisFlowMap} based on angle-restricted Steiner arborescences: Migration from Colorado and whisky exports from Scotland.}
\label{fig:flowmaps}
\end{figure}
In this paper we introduce a new variant of geometric minimal \emph{Steiner arborescences}, which captures the essential structure of flow trees and serves as a ``skeleton'' upon which to build high-quality flow trees. Our input consists of a point $r$, the \emph{root} (source), and $n$ points $t_1, \ldots, t_n$, the \emph{terminals} (targets). Visually appealing flow trees merge quickly, but smoothly. A geometric minimal Steiner arborescence on our input would result in the shortest possible tree, which naturally merges quickly. A Steiner arborescence for a given root and a set of terminals is a rooted directed \emph{Steiner tree}, which contains all terminals and where all edges are directed away from the root. Without additional restrictions on the edge directions (as in the rectilinear case or in the variant proposed in this paper), a geometric Steiner arborescence is simply a geometric Steiner tree with directed edges. However, Steiner arborescences have angles of $2\pi/3$ at every internal node and hence are quite far removed from the smooth appearance of hand-drawn flow maps. Our goal is hence to connect the terminals to the root with a Steiner tree of minimal length whose arcs obey a certain restriction on the angle they form with the root.

\begin{wrapfigure}[7]{r}{0.35\textwidth}
   \centering
   \includegraphics{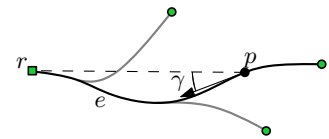}
   \small{\caption{The angle restriction. \label{fig:restriction}}}
\end{wrapfigure}
Specifically, we use a \emph{restricting angle} $\alpha < \pi/2$ to control the direction of the arcs of a Steiner arborescence $T$. Consider a point $p$ on an arc $e$ from a terminal to the root (see Figure~\ref{fig:restriction}). Let $\gamma$ be the angle between the vector from $p$ to the root $r$ and the tangent vector of $e$ at $p$. We require that $\gamma \leq \alpha$ for all points $p$ on $T$. We refer to a Steiner arborescence that obeys this angle restriction as \emph{angle-restricted Steiner arborescence}, or simply \emph{flux tree}.
Here and in the remainder of the paper it is convenient to direct flux trees from the terminals to the root. Also, to simplify descriptions, we often identify the nodes of a flux tree $T$ with their locations in the plane.%

In the context of flow maps it is important that flux trees can avoid obstacles, which model important features of the underlying geographic map. Furthermore, it is undesirable that terminals become internal nodes of a flux tree. We can ensure that our trees never pass directly through terminals by placing a small triangular obstacle just behind  each terminal (as seen from the root). Hence our input also includes a set of $m$ obstacles $B_1,\ldots, B_m$. We denote the total complexity (number of vertices) of all obstacles by $M$. In the presence of obstacles our goal is to find the shortest flux tree $T$ that is planar and avoids the obstacles.

The edges of flux trees are by definition ``thin'', but their topology and general structure are very suitable for flow trees. In a companion paper~\cite{InfoVisFlowMap} we describe an algorithm that thickens and smoothes a given flux tree while avoiding obstacles. Figure~\ref{fig:flowmaps} shows two examples of the maps computed with our algorithm, further examples and a detailed discussion of our maps can be found in~\cite{InfoVisFlowMap}.

\smallskip\noindent
{\bfseries Related work.} There is a multitude of related work on both the practical and the theoretical side of our problem and consequently we cannot cover it all.

One of the first systems for the automated creation of flow maps was developed by Tobler in the 1980s~\cite{FlowMapper,WaldoTobler1987}. His system does not use edge bundling and hence the resulting maps suffer from visual clutter. In 2005 Phan~\etal~\cite{Phan2005} presented an algorithm, based on hierarchical clustering of the terminals, which creates flow trees with bundled edges. This algorithm uses an iterative ad-hoc method to route edges and is often unable to avoid crossings. A second effect of this method is that flows are often routed along counterintuitive routes. The quality of the maps can be improved by moving the terminals, which, however, is considered to be confusing for users by cartography textbooks~\cite{Slocum2010}. Recent papers from the information visualization community explore alternative ways to visualize flows, by using multi-view displays~\cite{Guo2009}, animations over time~\cite{Boyandin2010}, or mapping techniques close to treemaps~\cite{Wood2010}.

There are many variations on the classic Steiner tree problem which employ metrics that are related to their specific target applications. Of particular relevance to this paper is the \emph{rectilinear Steiner arborescence} (RSA) problem, which is defined as follows. We are given a root (usually at the origin) and a set of terminals $t_1, \ldots, t_n$ in the northeast quadrant of the plane. The goal is to find the shortest rooted rectilinear tree $T$ with all edges directed away from the root, such that $T$ contains all points $t_1, \ldots, t_n$. For any edge of $T$ from $p = (x_p, y_p)$ to $q = (x_q, y_q)$ it must hold that $x_p \leq x_q$ and $y_p \leq y_q$. If we drop the condition of rectilinearity then we arrive at the \emph{Euclidean Steiner arborescence} (ESA) problem. In both cases it is  NP-hard~\cite{Shi2000,ss-rsap-05} to compute a tree of minimum length. Rao \etal~\cite{Rao92} give a simple $2$-approximation algorithm for minimum rectilinear Steiner arborescences. C\'{o}rdova and Lee~\cite{Cordova94} describe an efficient heuristic which works for terminals located anywhere in the plane. Ramnath~\cite{Ramnath03} presents a more involved $2$-approximation that can also deal with rectangular obstacles.  Finally, Lu and Ruan~\cite{Lu2000} developed a PTAS for minimum rectilinear Steiner arborescences, which is, however, more of theoretical than of practical interest.

Conceptually related are \emph{gradient-constrained minimum networks} which are studied by Brazil \etal~\cite{Brazil2001,Brazil2007} motivated by the design of
underground mines. Gradient-constrained minimum networks are minimum Steiner trees in three-dimensional space, in which the (absolute) gradients of all edges are no more than an upper bound $m$ (so that heavy mining trucks can still drive up the ramps modeled by the Steiner tree). Krozel~\etal~\cite{Krozel2006} study algorithms for turn-constrained routing with thick edges in the context of air traffic control. Their paths need to avoid obstacles (bad weather systems) and arrive at a single target (the airport). The union of consecutive paths bears some similarity with flow maps, although it is not necessarily crossing-free or a tree.

\smallskip\noindent
{\bfseries Results and organization.} In Section~\ref{sec:props} we derive properties of optimal (minimum length) flux trees. In particular, we show that they are planar and that the arcs of optimal flux trees consist of (segments of) logarithmic spirals and straight lines.
Flux trees have the \emph{shallow-light property}~\cite{Awerbuch1990}, that is, we can bound the length of an optimal flux tree in comparison with a minimum spanning tree on the same set of terminals and we can give an upper bound on the length of a path between any point in a flux tree and the root.
They also naturally induce a clustering on the terminals and smoothly bundle lines.
Unfortunately we can show
that it is NP-hard (Section~\ref{sec:npproof}) to compute optimal flux trees. Hence, in Section~\ref{sec:spiraltrees} we introduce a variant of flux trees, so called \emph{spiral trees}. The arcs of spiral trees consist only of logarithmic spiral segments. We prove that spiral trees approximate flux trees within a factor depending on the restricting angle $\alpha$. Our experiments show that $\alpha = \pi/6$ is a reasonable restricting angle, in this case the approximation factor is \mbox{$\sec(\alpha) \approx 1.15$}. In Section~\ref{sec:npproof} we show that computing optimal spiral trees remains NP-hard. For a special case, we give an exact algorithm in Section~\ref{sec:emptyregions} that runs in $O(n^3)$ time. In Section~\ref{sec:approximation} we develop a 2-approximation algorithm for spiral trees that works in general and runs in $O(n \log n)$ time. Finally, in Section~\ref{sec:obstacles} we extend our approximation algorithm (without deteriorating the approximation factor) to include ``positive monotone'' obstacles. On the way, we develop a new 2-approximation algorithm for rectilinear Steiner arborescences in the presence of positive monotone obstacles. Both algorithms run in $O((n+M) \log(n+M))$ time, where $M$ is the total complexity of all obstacles.

\section{Optimal flux trees} \label{sec:props}

Recall that our input consists of a root $r$, terminals $t_1, \ldots, t_n$, and a restricting angle $\alpha < \pi/2$. Without loss of generality we assume that the root lies at the origin. Recall further that an optimal flux tree is a geometric Steiner arborescence, whose arcs are directed from the terminals to the root and that satisfies the angle restriction. We show that the arcs of an optimal flux tree consist of line segments and parts of logarithmic spirals (Property~\ref{property:optedgeshape}), that any node except for the root has at most two incoming arcs (Property~\ref{property:optbinary}), and that an optimal flux tree is planar (Property~\ref{property:optplanar}). Finally, flux trees (and also spiral trees) have the shallow-light property (Property~\ref{property:shallowlight}).

\bigskip

\begin{wrapfigure}[10]{r}{.4\textwidth}
  \centering
  \includegraphics{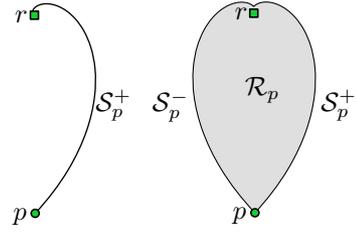}
  \small{\caption{Spirals and spiral regions.\label{fig:Spirals}}}
\end{wrapfigure}
\smallskip\noindent{\bfseries Spiral regions.} For a point $p$ in the plane, we consider the region $\mathcal{R}_p$ of all points that are \emph{reachable} from $p$ with an \emph{angle-restricted} path, that is, with a path that satisfies the angle restriction.
Clearly, the root $r$ is always in $\mathcal{R}_p$. The boundaries of $\mathcal{R}_p$ consist of curves that follow one of the two directions that form exactly an angle $\alpha$ with the direction towards the root. Curves with this property are known as \emph{logarithmic spirals} (see Figure~\ref{fig:Spirals}). Logarithmic spirals are self-similar; scaling a logarithmic spiral results in another logarithmic spiral. Logarithmic spirals are also self-approaching as defined by Aichholzer \etal~\cite{Aichholzer2001}, who give upper bounds on the lengths of (generalized) self-approaching curves. As all spirals in this paper are logarithmic, we simply refer to them as \emph{spirals}. For $\alpha < \pi/2$ there are two spirals through a point. The \emph{right spiral} $\mathcal{S}^{+}_p$ is given by the following parametric equation in polar coordinates, where $p = (R, \phi)$: $R(t) = R e^{-t}$ and $\phi(t) = \phi + \tan(\alpha) t$.
The parametric equation of the \emph{left spiral} $\mathcal{S}^{-}_p$ is the same with $\alpha$ replaced by $-\alpha$. Note that a right spiral $\mathcal{S}^{+}_p$ can never cross another right spiral $\mathcal{S}^{+}_q$ (the same holds for left spirals).
The spirals $\mathcal{S}^{+}_p$ and $\mathcal{S}^{-}_p$ cross infinitely often. The reachable region $\mathcal{R}_p$ is bounded by the parts of $\mathcal{S}^{+}_p$ and $\mathcal{S}^{-}_p$ with
$0 \leq t \leq \pi \cot(\alpha)$. We therefore call $\mathcal{R}_p$ the \emph{spiral region} of $p$. It follows directly from the definition that for all $q \in \mathcal{R}_p$ we have that $\mathcal{R}_q \subseteq \mathcal{R}_p$.

\begin{lemma}
\label{lem:shortestpath}
The shortest angle-restricted path between a point $p$ and a point $q \in \mathcal{R}_p$ consists of a straight segment followed by a spiral segment. Either segment can have length zero.
\end{lemma}
\medskip

\noindent{\bf Proof.\ }
Consider the spirals $\mathcal{S}^{+}_q$ and $\mathcal{S}^{-}_q$ through $q$, specifically the parts with $t \leq 0$ (see Figure~\ref{fig:Lemma1}). Any point on the opposite side of the spirals as $p$ is unable to reach $q$. Thus any shortest path from $p$ to $q$ cannot cross either of these spirals. If we see these spirals as obstacles and ignore the angle restriction for now, the shortest path $\pi$ is simply a straight segment followed by a spiral segment. Now consider any point $u$ on $\pi$. Because $\mathcal{S}^{+}_q$ and $\mathcal{S}^{+}_u$ cannot cross (same for $\mathcal{S}^{-}_q$ and $\mathcal{S}^{-}_u$), we get that $q \in \mathcal{R}_u$. Therefore $\pi$ also satisfies the angle restriction.\hfill\QED

\begin{figure}[b]
    \hfill
    \begin{minipage}[t]{.25\textwidth}
        \centering
        \includegraphics[height=1.2in]{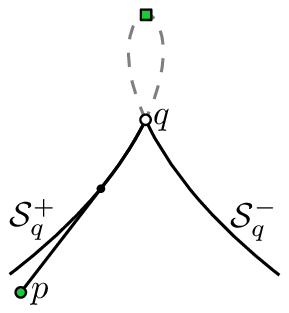}
        \small{\caption{Shortest path.\label{fig:Lemma1}}}
    \end{minipage}
    \hfill
    \begin{minipage}[t]{.45\textwidth}
        \centering
        \includegraphics[height=1.2in]{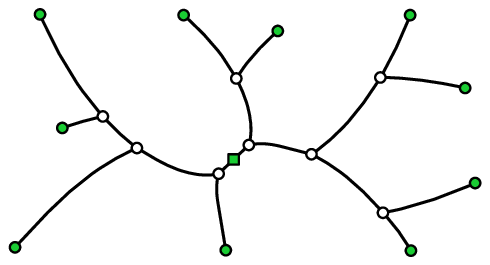}
        \small{\caption{An optimal flux tree ($\alpha = \pi/6$).\label{fig:OptFlowTree}}}
    \end{minipage}
    \hfill
    \begin{minipage}[t]{.25\textwidth}
        \centering
        \includegraphics[height=1.2in]{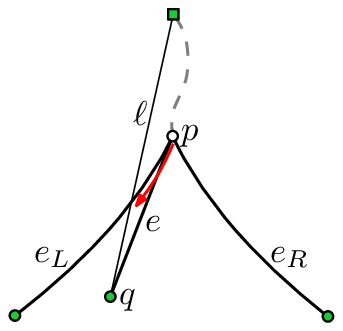}
        \small{\caption{Property 2.\label{fig:Property2}}}
    \end{minipage}
    \hfill\hfill
\end{figure}

\begin{property}
\label{property:optedgeshape}
An optimal flux tree consists of straight segments and spiral segments.
\end{property}
\medskip

\noindent{\bf Proof.\ }
Consider an optimal flux tree $T$. Now replace all edges between a terminal/the root and a Steiner node by the shortest angle-restricted path between the two points. This can only shorten $T$. By Lemma~\ref{lem:shortestpath}, the resulting flux tree consists of only straight segments and spiral segments. \hfill\QED

\begin{property}\label{property:optbinary}
Every node in an optimal flux tree $T$, other than the root $r$, has at most two incoming edges.
\end{property}
\medskip

\noindent{\bf Proof.\ }
Assume $T$ contains a node at $p$ with at least three incoming edges. Pick one of the incoming edges $e$ that is not leftmost or rightmost and let $q$ be the
other endpoint of $e$. Let $e_L$ and $e_R$ be the leftmost and rightmost incoming edges of $p$ (see Figure~\ref{fig:Property2}). Now consider the straight line $\ell$ from $q$ to the root $r$. Assume without loss of generality that $\ell$ passes $p$ on the left side (the right side is symmetric with $e_R$) or $\ell$ goes through $p$. We claim that we can locally improve the length of $T$ by moving the endpoint at $p$ of $e$ along $e_L$. The angle between $e$ and $e_L$ at $p$ is at most $\alpha$. Because $\alpha < \pi/2$ and because locally moving the endpoint of $e$ along $e_L$ will not make the spiral segment of $e$ longer, this will shorten the tree $T$. Also, locally moving the endpoint of $e$ along $e_L$ cannot suddenly violate the angle restriction (assuming that $\alpha > 0$). Contradiction.\hfill\QED

\begin{property}
\label{property:optplanar}
Every optimal flux tree is planar.
\end{property}
\begin{proof}
Assume two edges $e_1$ (from $p_1$ to $q_1$) and $e_2$ (from $p_2$ to $q_2$) cross. Let $u$ be the crossing between $e_1$ and $e_2$. Now simply remove the part of $e_1$ from $u$ to $q_1$. There is still a connection from $p_1$ to $r$ via $q_2$, so the resulting tree is still a proper flux tree. Also, removing a segment cannot violate the angle restriction and makes the tree shorter. Contradiction.
\end{proof}
The last property requires a more involved proof. We postpone the proof of this property until Section~\ref{sec:shallowlight}. Let $d^T(p)$ be the distance between $p$ and $r$ in a flux tree $T$ and let $d(p)$ be the Euclidean distance between $p$ and $r$.
\begin{property}
\label{property:shallowlight} The length of an optimal flux tree $T$ is at most $O((\sec(\alpha) + \csc(\alpha)) \log n)$ times the length of the minimum spanning tree on the same
set of terminals. Also, for every point $p \in T$, $d^T(p) \leq \sec(\alpha) d(p)$.
\end{property}

\section{Spiral trees}\label{sec:spiraltrees}

In this section we introduce \emph{spiral trees} and prove that they approximate flow trees. The arcs of a spiral tree consist only of spiral segments of a given $\alpha$ (see
Figure~\ref{fig:SpiralTree}). In other words, an optimal spiral tree is the shortest flow tree that uses only spiral segments. Spiral trees satisfy the angle restriction by
definition. Any particular arc of a spiral tree can consist of arbitrarily many spiral segments. That is, any arc of the spiral tree can switch between following its right spiral
and following its left spiral an arbitrary number of times. The length of a spiral segment can easily be expressed in polar coordinates. Let $p = (R_1, \phi_1)$
and $q = (R_2, \phi_2)$ be two points on a spiral, then the distance $D(p, q)$ between $p$ and $q$ on the spiral is
\begin{equation}
\label{eqn:spiraldist}
D(p, q) = \sec(\alpha) |R_1 - R_2| \, .
\end{equation}
Consider the shortest \emph{spiral path}---using only spiral segments---between a point~$p$ and a point $q$ reachable from $p$. The reachable region for $p$ is still its spiral region $\mathcal{R}_p$, so necessarily $q \in \mathcal{R}_p$. The length of a shortest spiral path is given by Equation~\ref{eqn:spiraldist}. The shortest spiral path is not unique, in particular, any sequence of spiral segments from $p$ to $q$ is shortest, as long as we move towards the root.

\begin{figure}[b]
    \hfill
    \begin{minipage}[t]{.45\textwidth}
        \centering
        \includegraphics[width=.8\textwidth]{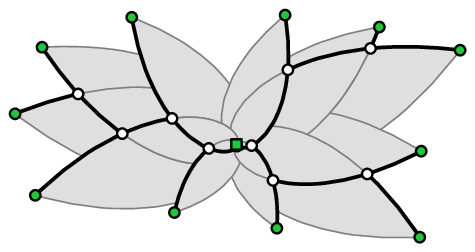}
        \caption{Spiral tree with spiral regions.}
        \label{fig:SpiralTree}
    \end{minipage}
    \hfill
    \begin{minipage}[t]{.5\textwidth}
        \centering
        \includegraphics{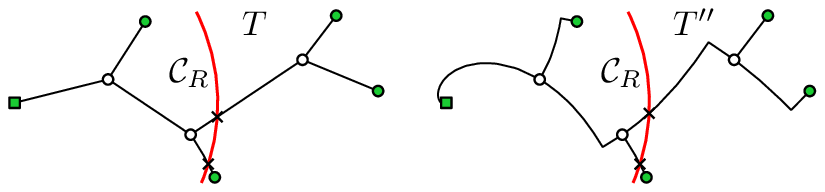}
        \caption{$T$ and $T''$.}
        \label{fig:Theorem1}
    \end{minipage}
    \hfill\hfill
\end{figure}

\begin{theorem}
\label{thm:fluxapprox}
The optimal spiral tree $T'$ is a $\sec(\alpha)$-approximation of the optimal flux tree $T$.
\end{theorem}
\begin{proof}
Let $\mathcal{C}_R$ be a circle of radius $R$ with the root $r$ as center. A lower bound for the length of $T$ is given by $L(T) \geq \int_0^\infty |T \cap \mathcal{C}_R| dR$, where $|T \cap \mathcal{C}_R|$ counts the number of intersections between the tree $T$ and the circle $\mathcal{C}_R$. Using Equation~\ref{eqn:spiraldist}, the length of $T'$ is $L(T') = \sec(\alpha) \int_0^\infty |T' \cap \mathcal{C}_R| dR$. Now consider the spiral tree $T''$ with the same nodes as $T$, but where all arcs between the nodes are replaced by a sequence of spiral segments (see Figure~\ref{fig:Theorem1}). For a given circle $\mathcal{C}_R$, this operation does not change the number of intersections of the tree with $\mathcal{C}_R$, i.e. $|T \cap \mathcal{C}_R| = |T'' \cap \mathcal{C}_R|$. So we get the following:
\[
L(T') \leq L(T'') = \sec(\alpha) \int_0^\infty |T'' \cap \mathcal{C}_R| dR = \sec(\alpha) \int_0^\infty |T \cap \mathcal{C}_R| dR \leq \sec(\alpha) L(T)
\]
\end{proof}
Next to the fact that optimal spiral trees are a good approximation of optimal flux trees, they also maintain important properties of optimal flux trees, namely Properties~\ref{property:optbinary} and~\ref{property:optplanar}.
%

\begin{lemma}
\label{lem:spiralplandeg}
An optimal spiral tree is planar and every node, other than the root, has at most two incoming edges. The root has exactly one incoming edge.
\end{lemma}
\begin{proof}
First of all, only two spirals go through a single point: the left and the right spiral. So every node other than the root $r$ has at most two incoming arcs, otherwise there is a repeated spiral segment which can be removed. Furthermore, the same arguments as in the proof of Property~\ref{property:optplanar} yield that also the optimal spiral tree is planar.
\end{proof}

When we approximate an optimal flux tree by a spiral tree, we can further reduce the length of the tree by replacing every arc by the shortest angle-restricted path between its endpoints.
This operation does not improve the approximation factor, but it improves the tree visually.

\subsection{Shallow-Light Property}\label{sec:shallowlight}
In the following we prove the shallow-light property for optimal spiral trees. That is, we bound the length of an optimal spiral tree in comparison with a minimum spanning tree on the same set of terminals and we give an upper bound on the length of a path between any point in a spiral tree and the root. Since for flux trees the length of such paths and the total length are not larger, we can conclude that optimal flux trees also have the shallow-light property (Property~\ref{property:shallowlight}).
%
The second part of the property (shallowness) is easy to see.
The path of a node to the root in any spiral tree is by Equation~\ref{eqn:spiraldist}
bounded by $\sec(\alpha)$ times the distance of the node to the root. 

We now show that the length of an optimal spiral tree approximates the length of a minimum spanning tree by a factor of $O((\sec(\alpha) + \csc (\alpha)) \log n )$. We build a spiral tree in the following way.
 First we find a short cycle through the points. We then take a matching based on this cycle and pairwise join points by spiral segments. This results in $\lceil n/2 \rceil$ components. On these we again find a matching and pairwise join them and so on. We need to ensure that the set of spiral segments used in the construction is compatible with a spiral tree.

 Throughout this section we will assume that the root of the tree is placed at the origin. We call a sequence of spiral segments between two points \emph{inward going} if the distance from the segments to the origin have no local maximum except possibly at the two points. In particular, if a point is in the spiral region of another, a path of decreasing distance to the root from the outer point to the inner one would be inward going. Any pair of points can be joined by an inward going sequence of two spiral segments and inward going sequences are compatible with spiral trees if we use the point with smallest distance to the root as join node.

We need to bound the length of such a sequence. For this we first bound the length of a spiral segment. Let $p$, $p'$ be two points on a spiral segment with polar coordinates $p = (R,\phi)$ and $p' = (R',\phi')$. Equation~\ref{eqn:spiraldist} gives us a bound in terms of $R, R'$. To bound the length in terms of $R,\phi,\phi'$, let us assume $R' \leq R$. The other case is analogous. We consider the parametric equations of the spiral through $p$ and $p'$ with $(R(0),\phi(0)) = (R, \phi)$. For $p'$ we obtain the equations $R' = R e^{-t}$ and $\phi' = \phi + \tan (\alpha) t$ (or $\phi' = \phi - \tan (\alpha) t$ depending on whether the points lie on a left or right spiral). Solving for $R'$ yields
$R' = R e^{- (|\phi'-\phi|)/\tan (\alpha)}$.
Inserting this into Equation~\ref{eqn:spiraldist} gives
\begin{equation}
\label{eqn:spiral2}
D(p,p') = \sec (\alpha) R (1-e^{- (|\phi'-\phi|)/\tan (\alpha)})\, .
\end{equation}
Equation~\ref{eqn:spiral2} has several consequences. Given two points $p$, $q$ that do not lie in the spiral regions of each other. Assume we have two sequences of spiral segments, each connecting $p$ and $q$ such that no ray through the origin intersect a sequence twice, and such that parameterized by the angle $\phi$ to the origin, one sequence has a smaller or equal distance to the origin for all $\phi$. Then sweeping over the angle and summing up the contributions of Equation~\ref{eqn:spiral2} gives that the sequence closer to the origin has a smaller (or equal) total arc length. Thus, the shortest connection between $p$ and $q$ is obtained by simply joining $p$ and $q$ by an inward going sequence of two spiral segments. 

Another consequence of Equation~\ref{eqn:spiral2} is the following. Again $p = (R, \phi)$, $q = (R', \phi')$ are two points that do not lie in the spiral regions of each other. Further assume $R, \phi, \phi'$ are given, but for $R'$ we only know $R' \leq R$. Then the arc length of the inward going sequence of two spiral segments joining $p$ and $q$ (using the angle range between $\phi$ and $\phi'$) is maximized for $R=R'$. This follows from the same argument as above, i.e., the resulting
sequence of spiral segments dominates all others in terms of distance to the origin.

So far we have not linked the arc length of the spiral segments between two points with the Euclidean distance between the points. We do this by the following lemma.
\begin{lemma}
Two points in the plane of distance $D$ can be connected by an inward going path of logarithmic spirals of angle $\alpha$ such that the summed length of the spiral segments is bounded by $3 D \max (\sec (\alpha), \csc (\alpha))$. The path uses at most two spiral segments.
\end{lemma}
\begin{proof}
Let $p_1$, $p_2$ be two points of distance $D$ with polar coordinates $p_1 = (R_1,\phi_1)$ and $p_2 = (R_2,\phi_2)$. Without loss of generality we assume that $R_1 \leq R_2$, $\phi_1 \leq \phi_2$ and $\phi_2 - \phi_1 \leq \pi$. We first handle the case that $p_1$ lies in the spiral region of $p_2$. In this case we can connect the points by an inward going path from $p_2$ to $p_1$ using two spiral segments. By Equation~\ref{eqn:spiraldist} the length of this path is $\sec(\alpha) (R_2-R_1) \leq \sec(\alpha) D$, which proves the claim for this case.

Next we handle the case that $p_1$ does not lie in the spiral region of $p_2$. In this case we join the points using the right spiral through $p_1$ and the left spiral through $p_2$. Let $p = (R, \phi)$ be the point where the two points first join. The summed length of the spiral segments is $L = \sec (\alpha) (R_1 + R_2 - 2R)$, which we need to bound in terms of $D$

We distinguish two cases. First assume the points have a distance of at most $3D/2$ to the root. Then we obtain the connection between the points by simply connecting both to the root. Then $L \leq \sec(\alpha) (R_1+R_2) \leq 2 \sec(\alpha) 3D/2 = 3 D \sec(\alpha)$.
Next assume $R_2>3D/2$.
From the discussion of Equation~\ref{eqn:spiral2} above we know that $L$ is maximized for $R_1=R_2$. In this case we have $\phi = (\phi_2+\phi_1)/2$ and therefore
\begin{align*}
L &= \sec (\alpha) (2 R_2 - 2 R_2 e^{- \frac{\phi_2-\phi_1}{2\tan (\alpha)}})  = \sec (\alpha) 2 R_2 (1 - e^{- \frac{\phi_2-\phi_1}{2\tan (\alpha)}})\\
 &\leq \sec (\alpha) 2 R_2 \frac{\phi_2-\phi_1}{2\tan (\alpha)} = \csc (\alpha) R_2 (\phi_2-\phi_1).
\end{align*}
%
%
It remains to bound $R_2 (\phi_2-\phi_1)$ in terms of the Euclidean distance $D$ of the two points. Observe that for given $p_2$ (with $R_2>3D/2$) and $D$ the angle $\phi_2-\phi_1$ is maximized if the line through the origin and $p_1$ is tangent to the circle of radius $D$ around $p_2$. Thus $\phi_2-\phi_1$ is maximized if the angle formed by $p_2$, $p_1$ and the origin is $\pi/2$. In this case $\phi_2-\phi_1 = \arcsin(D/R_2)$. Thus, in general
\[
\phi_2-\phi_1 \leq \arcsin \left( \frac{D}{R_2} \right) =  \arctan \left( \frac{D}{\sqrt{R_2^2 - D^2}} \right) \leq \frac{D}{\sqrt{R_2^2 - D^2}} \, .
\]
Since $R_2>3D/2$, we have $\sqrt{R_2^2 - D^2}\geq R_2 \sqrt{1-4/9}$. Plugging this into the above bound gives $\phi_2-\phi_1 \leq D/(R_2\sqrt{5/9})$. Now inserting this into the bound on $L$ gives
\[
L \leq \csc (\alpha) D/\sqrt{5/9} < 3 \csc (\alpha) D.
\]
Combining the cases results in the claimed bound.
\end{proof}

%

\begin{theorem}
The length of the optimal spiral tree of a set of points is bounded by \\$3 \lceil \log_2 n \rceil \max (\sec (\alpha), \csc (\alpha ))$ times the length of the minimum spanning tree of the set of points with the origin included.
\end{theorem}
\begin{proof}
In the following we construct a spiral tree for which this bound holds. Let $L$ be the length of the minimum spanning tree on the points including the origin. Let $\kappa = 3 \max (\sec (\alpha), \csc (\alpha ))$. Let $C_1$ be a cycle through the points of length at most $2L$ (e.g., obtained by ordering the points based on a depth first search in the minimum spanning tree). We replace each edge of $C_1$ by an inward going sequence of at most two spiral segments. This results in a cycle $C'_1$ of sequences of spiral segments of length at most $2\kappa L$. By taking either every even or every odd sequence we join pairs (possibly leaving the root unmatched) of nodes by spiral segments of total length at most $\kappa L$. We repeat the construction on the join nodes (and possibly an unmatched point)) using the cycle $C_2$ induced by the order given by $C'_1$. Again we form a cycle of spiral segments through the vertices of $C_2$. If we parameterize corresponding sequences in the cycles $C'_1$ and $C'_2$ by the angle $\phi$ from the origin, the sequences in cycle $C'_2$ are closer to the origin than the corresponding sequences in $C'_1$ for all angles. Thus as consequence of Equation~\ref{eqn:spiral2} the length of $C'_2$ is bounded by the length of $C'_1$ and therefore by $2\kappa L$. Thus as in the previous step we can join pairs of join nodes using spiral segments of total length at most $\kappa L$. Next we construct $C'_3$ from $C'_2$ in the same way and iterate the construction. After at most $\lceil \log_2 n \rceil$ all nodes have been joined. The total length is then $\lceil \log_2 n \rceil \kappa L$ as claimed.
\end{proof}

\noindent
\subsection{Relation with rectilinear Steiner arborescences.} \label{sec:transformation}
Both rectilinear Steiner arborescences and spiral trees contain directed paths, from the root to the terminals or vice versa. Every edge of a rectilinear Steiner arborescence is restricted to point right or up, which is similar to the angle restriction of flux and spiral trees. In fact, there exists a transformation from rectilinear Steiner arborescences into spiral trees. Consider the following transformation from the coordinates
$(x, y)$ of a rectilinear Steiner arborescence to the polar coordinates $(R, \phi)$ of a spiral tree.
\begin{align}
\label{eqn:transformation}
R    =& e^{x + y} \\
\nonumber \phi =& (y - x) \tan(\alpha)
\end{align}
Assume we keep one of the coordinates $x$ or $y$ fixed. Using the spiral equation from Section~\ref{sec:props} we see that the result is a spiral. More specifically, keeping $x$ fixed results in left spirals and keeping $y$ fixed results in right spirals. So that means that the above transformation transforms horizontal and vertical lines to right and left spirals, respectively (see Figure~\ref{fig:RSATransform}). The transformation maps the root of the rectilinear Steiner arborescence to $(1, 0)$. Thus, to get a valid spiral tree, we still need to connect $(1, 0)$ to $r$.

\begin{lemma}
The transformation in Equation~\ref{eqn:transformation} transforms a rectilinear Stei\-ner arborescence into a spiral tree.
\end{lemma}
Unfortunately, the transformation has several shortcomings. First of all, the transformation is not a bijection, it is a surjection. That means we can invert the transformation,
but only if we restrict the domain in the rectilinear space. But most importantly, the metric does not carry over the transformation. That means that it is not necessarily true
that the minimum rectilinear Steiner arborescence transforms to the optimal spiral tree. Thus the relation between the concepts cannot be used directly and algorithms developed for rectilinear Steiner arborescences cannot be simply modified to compute spiral trees. However, the same basic ideas can often be used in both settings.

\begin{figure}[t]
  \centering
  \includegraphics[width=.6\textwidth]{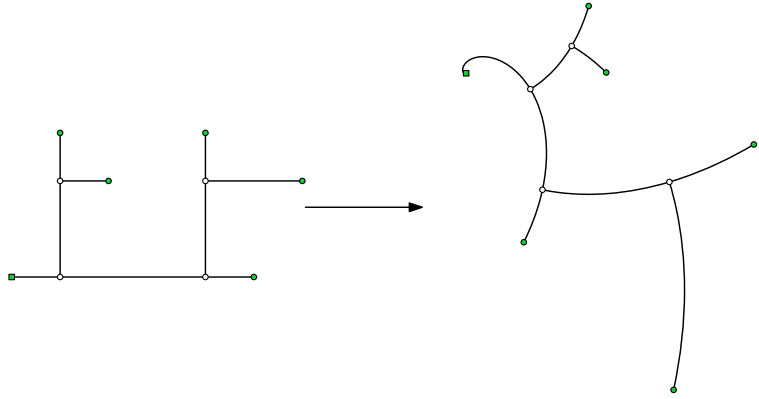}
  \caption{A rectilinear Steiner arborescence transformed to a spiral tree.}
  \label{fig:RSATransform}
\end{figure}

\section{Computing spiral trees} \label{sec:computespirals}

In this section we describe algorithms to compute (approximations of) optimal spiral trees. First we show that it is NP-hard to compute optimal flux or spiral trees. Then we give an exact algorithm for computing optimal spiral trees in the special case that all spiral regions are empty, i.e. $t_i \notin \mathcal{R}_{t_j}$ for all $i \neq j$. Finally we give an approximation algorithm for computing optimal spiral trees in the general case.

\subsection{Computing optimal flux and spiral trees is NP-hard}\label{sec:npproof}

For the hardness proofs we will choose $\alpha = \pi/4$. The reduction is from the rectilinear Steiner arborescence (RSA) problem~\cite{ss-rsap-05} for spiral trees, and
from the Euclidean Steiner arborescence (ESA) problem~\cite{Shi2000} for flux trees.
Shi and Su~\cite{ss-rsap-05} proved by a reduction from planar 3SAT that the decision versions of the RSA problem and the ESA problem are NP-hard. Their reduction uses points on an $O(m \times m)$ grid, where $m$ bounds the size of the 3SAT instance. We can assume the grid to be an integer grid. Now if there is a satisfying assignment, the optimal RSA and ESA have an integer length $K$, while if there is no such assignment the optimal RSA and ESA have length at least $K+1$.

We can therefore state the problems for which they proved NP-hardness and from which we will reduce as follows.
\\
\noindent {\bf Instance: } A set of integer points $P = \{ p_1,\ldots, p_N \}$ in the first quadrant of the plane with coordinates bounded by $O(N^2)$; a positive integer $K$.

\noindent {\bf Question (RSA): } Is there a RSA of total length $K$ or less? Otherwise the shortest RSA has length at least $K+1$.

\noindent {\bf Question (ESA): } Is there a ESA of total length $K$ or less? Otherwise the shortest RSA has length at least $K+1$.

\begin{wrapfigure}[10]{r}{.25\textwidth}
  \centering
  \vspace{.35\baselineskip}
  \includegraphics{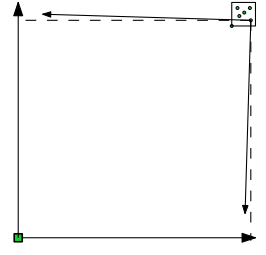}
  \vspace{-.35\baselineskip}
  \small{\caption{Reduction.\label{fig:NPhard-idea}}}
\end{wrapfigure}
The basic idea is sketched in Figure~\ref{fig:NPhard-idea}. Assume we are given an instance of the Euclidean Steiner arborescence problem with polynomially bounded coordinates. We translate the set of terminals by a large polynomial factor along the diagonal with slope 1 and place a new root at the origin. If the bound on the coordinates is small (the square in Figure~\ref{fig:NPhard-idea}) relative to the factor of the translation, then the angle formed by the line through any of the translated points and the origin with the x-axis is ``more or less $\pi/4$''. A $\pi/4$-restricted flux tree thus behaves within this square ``more or less'' like an Euclidean Steiner arborescence. For spiral trees we use the same setup but show that the distances on the spiral tree approximate the $L_1$-norm. However, quantifying ``more or less" precisely is technically rather involved and will be done in this section.

To draw the connection from Steiner arborescences to flux and spiral trees we generalize the concept of RSAs and ESAs. For flux and spiral trees the angle $\alpha$ is bounded relative to the root while for RSAs and ESAs the angle that an edge can make with the $x$-axis is bounded (or with any given line through the origin). For ESAs the angle of an edge is in $[0,\pi/2]$, while for RSAs the angle is in $\{0, \pi/2 \}$. We call a Steiner arborescence with angles in $[\beta,\pi/2-\beta]$ a $(\geq \beta)$-Steiner arborescence ($(\geq \beta)$-SA) and a Steiner arborescence with angles in $\{\beta,\pi/2-\beta\}$ a $\beta$-Steiner arborescence ($\beta$-SA). We do not restrict $\beta$ to be positive but to $-\pi/4 < \beta < \pi/4$. In the following Steiner arborescences are typically not rooted at the origin.

Now let $-\pi/4 < \beta < \beta' < \pi/4$. Every $(\geq \beta')$-SA is a $(\geq \beta)$-SA but the converse does not hold.
However, we can transform a $(\geq \beta)$-SA to $(\geq \beta')$-SA of similar length. Our transformation first transforms the whole tree and then connects the original points to their images under the first transformation. The transformation actually changes the location of the root. For the reduction we give this is not a problem because in the reduction we will have an additional root to which both the root of the original tree and the root of the transformed tree need to connect.

Let $\eta_1 = (-\cos \beta, -\sin \beta)$ and $\eta_2 = (-\sin \beta, -\cos \beta)$. Let $p_1, \ldots, p_n$ be points with $p_i = u_i \eta_1 + v_i \eta_2$, $(u_i,v_i) \in [0,B]^2$, where $B$ may depend on $n$. Let $T$ be a $(\geq \beta)$-SA on $p_1, \ldots, p_n$ with root $B \eta_1 + B \eta_2$.
Let $\lambda = \cos ( \beta + \beta')/ \cos (2 \beta')$ and $\eta'_1 = \lambda (-\cos \beta', -\sin \beta')$ and $\eta'_2 = \lambda (-\sin \beta', -\cos \beta')$.
We transform $T$ by the following transformation $\tau \colon \mathbb{R}^2 \rightarrow \mathbb{R}^2$: Any point $p = u \eta_1 + v \eta_2$ is mapped to $q = u \eta'_1 + v \eta'_2$. We obtain a Steiner arborescence $T'$ on $p_1, \ldots, p_n$ with root $B \eta'_1 + B \eta'_2$ by connecting $p_i$ to $\tau (p_i)$ by a line segment. Let $\lambda' = \sin (\beta' - \beta)/ \cos (\beta + \beta')$.
\begin{lemma}\label{lem:transtree}
$T'$ is a $(\geq \beta')$-SA and
$
|T'| \leq \lambda |T| + 2 \lambda' B n.
$
\end{lemma}

\noindent{\bf Proof.} If we ignore the connections between the $p_i$s and $\tau(p_i)$s the resulting transformed tree by construction fulfils the angle restriction and its length is $\lambda |T|$. We therefore only need to show that the connections fulfil the angle restriction and that the length of any connection is bounded by $\lambda'$. We have $p_i - \tau(p_i) = u_i (\eta_1-\eta'_1) + v_i (\eta_2-\eta'_2)$. It therefore suffices to prove that $\eta_1-\eta'_1$ and $\eta_2-\eta'_2$ fulfil the angle restriction. Since $\eta_2-\eta'_2$ is $\eta_1-\eta'_1$ mirrored at the diagonal with slope 1, it actually suffices to consider $\eta_1-\eta'_1$.

\begin{wrapfigure}[8]{r}{.25\textwidth}
  \centering
  \includegraphics{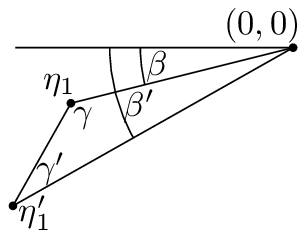}
  \small{\caption{Lemma~\ref{lem:transtree}.\label{fig:LemmaTransformation}}}
\end{wrapfigure}

Consider the triangle formed by the origin, $\eta_1$ and $\eta'_1$ (see Figure~\ref{fig:LemmaTransformation}). We have $|\eta_1| = 1$ and $|\eta'_1|=\lambda$. By the law of sines $\lambda = \sin \gamma / \sin \gamma'$. This equation holds for $\gamma = \pi/2 + \beta + \beta'$, since then $\gamma' = \pi - (\beta'-\beta) - \gamma = \pi/2 - 2 \beta'$ and therefore $\sin \gamma / \sin \gamma' = \sin (\pi/2 + \beta + \beta') / \sin (\pi/2 - 2 \beta') = \cos (\beta + \beta')/\cos (2 \beta') = \lambda$. On the other hand with $\gamma = \pi/2 + \beta + \beta'$ we have $\eta'_1$ is indeed reachable from $\eta_1$ in a  $(\geq \beta')$-SA. The length of the connection is here the length of the third side of the triangle, which is $\sin (\beta' - \beta)/\sin \gamma = \lambda'$. More generally the length of a connection is bounded by $2B$, that is $B$ for each coordinate. Since we have $n$ such connections the bound of the lemma holds. \hfill\QED
\medskip

\noindent
In Lemma~\ref{lem:transtree} we have two summands, one depending on $|T|$ and one on $n$. Since the terminals lie on an integer grid and since every terminal has to connect to the tree, the length of the tree is at least of order $n$.
\begin{obs}\label{obs:bound-on-n}
If the terminals of a Steiner arborescence T have integer coordinates then $n \leq 2|T|$.
\end{obs}

\begin{theorem}\label{thm:nphard-flux}
It is NP-hard to compute the optimal flux tree of a point set.
\end{theorem}
\noindent{\bf Proof.}
Given an ESA instance with root $(0,0)$ and with the coordinates $x$ and $y$ of any point $(x,y)$ on the tree bounded by $c n^2$, we translate every terminal by $2 n^k (1,1)$ for a constant integer $k>2$ specified later. We include the translated root in the point set but not as root. Instead we take $(0,0)$ again as root 
The shortest ESA is simply the originally shortest translated with one additional edge from $(2 n^k,2 n^k)$ to $(0,0)$. Now, consider a point $(c' n^k + x, c' n^k + y)$ with $c'\geq 1$ and $0\leq x,y \leq c n^2$. The angle of a line through this point and the origin with the diagonal of slope~1 is bounded by $\beta_{\max} = c n^2/n^k = c/n^{k-2}$. Now, restricted to such points every $(\geq \beta_{\max})$-SA is a flux tree with $\alpha = \pi/4$, and every such flux tree a $(\geq -\beta_{\max})$-SA. Also, every $(\geq \beta_{\max})$-SA is a Euclidean Steiner arborescence, and every Euclidean Steiner arborescence a $(\geq  -\beta_{\max})$-SA. Thus, if we show that $(\geq \beta_{\max})$-SAs approximate $(\geq  -\beta_{\max})$-SAs well, this directly implies that flux trees approximate Euclidean Steiner arborescences well. More specifically, if we want to show that the length of the shortest flux tree approximates the shortest Euclidean Steiner arborescence up to a precision of $1$ (so that we can make the distinction between $K$ and $K+1$) then it is sufficient to prove that a $(\geq \beta_{\max})$-SA $T'$ can approximate a $(\geq  -\beta_{\max})$-SA $T$ up to this precision.

By Lemma~\ref{lem:transtree} and Observation~\ref{obs:bound-on-n} we get $|T'| \leq \lambda |T| + 2 \lambda' c n^2 n \leq |T| (\lambda + 4 c \lambda' n^2)$. Now, $\lambda = 1/ \cos (2 \beta_{\max}) < 1 / (1 - c/n^{k-2})$ and $\lambda' = \sin (2\beta_{\max}) < 2\beta_{\max} = c/n^{k-2}$. Thus,
\[
\lambda + 4 c \lambda' n^2 < 1 / (1 - c/n^{k-2}) + 8c^2 /n^{k-4} = 1 + o(1/n^4)
\]
for $k>8$. Since $|T| = O(n^4)$, this allows us to approximate the length up a o(1)-term. Note that we still need to connect the root of $|T|$ and the root of $|T'|$ to $(0,0)$. The length of this connection is slightly different because the roots of the trees are different, but the difference is negligible compared to the difference of $|T'|$ and $|T|$. \hfill\QED
\medskip

It remains to prove NP-hardness for spiral trees.
\begin{theorem}\label{thm:nphard-spiral}
It is NP-hard to compute the optimal spiral tree of a point set.
\end{theorem}
\noindent{\bf Proof.}
We use the same construction as above for spiral trees but we start with a rectilinear Steiner tree instance instead of a Euclidean Steiner tree instance. To adapt the reduction it suffices to show that within the relevant part of the tree, that is the part in $[2n^k, 2n^k + c n^2]$, the length of a spiral segment between two points $p$, $q$ is up to a small error the same as $c''\|p-q\|_1$ for a suitable constant factor $c''$. The length of the spiral is $D(p,q) = \sec (\pi/4) ||p|-|q|| = \sqrt{2} ||p|-|q||$. Assume $x_q \leq x_p$ and $y_q \leq y_p$. Let $q' = (x_p, y_q)$. We have $| |p|-|q| | = (|p|-|q'|)+(|q'|-|q|)$. The difference $|p|-|q'|$ measures how much the distance to the origin decreases while moving from $p$ to $q'$. Let $y_0 = y_p-y_q$ be the length of the line segment between $p$ and $q'$ and let $\sigma \colon [0,y_0] \rightarrow \mathbb{R}^2$ be this line segment parameterized uniformly. For a point $u$ let $\gamma (u)$ be the angle formed by the line through the origin and $u$ with the $x$-axis. We have $|p|-|q'|=\int_0^{y_0} \cos \gamma (\sigma (u)) \mathrm{d}u$. Now $\pi/4 - \beta_{\max} \leq \gamma (\sigma (u)) \leq \pi/4 + \beta_{\max}$ and therefore $1/\sqrt{2} - \beta_{\max} \leq \cos \gamma (\sigma (u)) \leq 1/\sqrt{2} + \beta_{\max}$. Thus $||p|-|q'| - y_0/\sqrt{2}| \leq c/n^{k-2}$. By the same argument we have that $||q'|-|q| - x_0/\sqrt{2}| \leq c/n^{k-2}$, where $x_0 = x_p-x_q$. Therefore,
\begin{align*}
|D(p,q) - \|p-q\|_1| &= |\sqrt{2} ((|p|-|q'|)+(|q'|-|q|))- \|p-q\|_1| \\
                     &\leq |\sqrt{2} (x_0/\sqrt{2} + y_0/\sqrt{2} + 2c/n^{k-2}) - \|p-q\|_1|\\
                     &= 2\sqrt{2}c/n^{k-2}.
\end{align*}
Since the length of the RSA instance is in $O(n^4)$, the difference between measuring the length of a spiral segment versus taking the $L_1$-distance of endpoints of segments is in $o(1)$ for $k>4$. Thus, computing the optimal spiral tree is NP-hard.
\hfill\QED
\medskip

To prove NP-hardness it was sufficient to consider one value of $\alpha$, namely $\alpha = \pi/4$. Nonetheless, it is an interesting problem whether with results also holds for a given smaller $\alpha$. We believe that the NP-hardness proof in~\cite{ss-rsap-05} can be adapted to $\beta$-SAs and $(\geq \beta)$-SAs for $0<\beta \leq \pi/4$. With this we could also generalize our result to smaller $\alpha$.

\subsection{Optimal spiral trees with empty spiral regions} \label{sec:emptyregions}

Assume we are given an input instance such that $t_i \notin \mathcal{R}_{t_j}$ for all $i \neq j$. We give an exact polynomial time algorithm that computes optimal spiral trees for input instances with this property.

Before we discuss the algorithm, we first give a structural result on optimal spiral trees for these special instances. Assume all terminals are ordered radially (on angle) in counterclockwise direction around $r$ and are numbered as such. This means that the first terminal $t_1$ is arbitrary and the remaining terminals $t_2, \ldots, t_n$ follow this order. First note that, for these instances, every terminal is a leaf in any spiral tree. That is because no terminal can be reached by another terminal, so no terminal can have incoming edges. More important is the following result.

\begin{lemma}
\label{lem:DPleaforder}
If the spiral regions of all terminals are empty, then the leaf order of any planar spiral tree follows the radial order of the terminals.
\end{lemma}
\begin{proof}
Assume this is not case, so that the leaf order skips leafs $t_i, \ldots, t_j$, or in other words jumps from $t_{i-1}$ to $t_{j+1}$. Pick any terminal $t_k$ with $i \leq k \leq j$. Let $\pi$ be the path in the spiral tree from $t_{i-1}$ to $t_{j+1}$. Because the leaf order jumps from $t_{i-1}$ to $t_{j+1}$, no leaf is connected to the outside (as seen from $r$) of $\pi$. However, because $t_k \notin \mathcal{R}_{t_{i-1}}$, $t_k \notin \mathcal{R}_{t_{j+1}}$ and $t_{i-1}, t_{j+1} \notin \mathcal{R}_{t_{k}}$, the path $\pi$ cuts through $\mathcal{R}_{t_{k}}$ separating $t_k$ from $r$. So the only way for $t_k$ to be connected to $r$ is to cross $\pi$. Contradiction.
\end{proof}

\begin{cor}
If the spiral regions of all terminals are empty, then the leaf order of the optimal spiral tree follows the radial order of the terminals.
\end{cor}

\begin{figure}[b]
  \centering
  \includegraphics{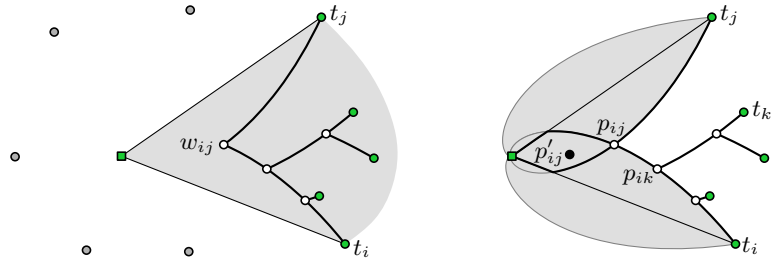}
  \caption{Left: The wedge $w_{ij}$ for terminals $t_i, \ldots, t_j$. Right: $p_{ij}$ is the optimal join point.}
  \label{fig:DPWedge}
\end{figure}

\noindent Using the above lemma we can use a simple dynamic programming algorithm to compute the optimal spiral tree. We simply solve all subproblems that ask for the optimal spiral subtree for a sequence of terminals $t_i, \ldots, t_j$. We require that this subtree is contained in the unbounded wedge $w_{ij}$ from the radial line through $t_i$ to the radial line through $t_j$ (see Figure~\ref{fig:DPWedge} left). Define $p_{ij}$ as the intersection of $\mathcal{S}^{+}_{t_i}$ and $\mathcal{S}^{-}_{t_j}$ ($p_{ii} = t_i$). As every internal node has exactly two incoming edges (Observation~\ref{lem:spiralplandeg}), we split the subtree into two subtrees at every internal node. To compute the optimal spiral tree for a sequence of terminals $t_i, \ldots, t_j$, we simply compute the optimal way to split the subtree into two subtrees by trying all possibilities. We then connect both subtrees to $p_{ij}$. Note that, by Lemma~\ref{lem:DPleaforder}, we need to check only $j - i$ ways to split this subtree. If $F(i, j)$ is the length of the optimal spiral subtree for the terminals $t_i, \ldots, t_j$ (contained in $w_{ij}$), then we can perform dynamic programming using the following recursive relation.

\begin{equation}
F(i, j) = \begin{cases} 0, & \mbox{if } i = j, \\ \min_k (F(i, k) + F(k+1, j) + D(p_{ik}, p_{ij}) + D(p_{(k+1)j}, p_{ij})), & \mbox{otherwise.} \end{cases}
\end{equation}

\noindent Note that we allow $j < i$, because we have a cyclical order. However, the value of $k$ in the above equation must be between $i$ and $j$ in the cyclical order. The distance function $D$ is defined as in Equation~\ref{eqn:spiraldist}.

\begin{lemma}
The function $F(i, j)$ describes the length of the optimal spiral subtree for the terminals $t_i, \ldots, t_j$ contained in $w_{ij}$.
\end{lemma}
\begin{proof}
We prove the lemma by induction. If $i = j$, then $F(i, j) = 0$ is clearly correct. If $i \neq j$, then, by Lemma~\ref{lem:DPleaforder}, we compute the minimum of all possible splits for the corresponding subtree. Let this split be between $t_k$ and $t_{k+1}$. By induction, $F(i, k)$ and $F(k+1, j)$ describe the lengths of the optimal subtrees. We need to show that $p_{ij}$ is the optimal point to join the subtrees. For the sake of contradiction, assume the optimal join point is $p'_{ij}$. This point must be in the intersection of $\mathcal{R}_{t_i}$, $\mathcal{R}_{t_j}$ and $w_{ij}$ (see Figure~\ref{fig:DPWedge} right). This means that $p'_{ij} \in \mathcal{R}_{p_{ij}}$. We can replace the edges $p_{ik} \rightarrow p'_{ij}$ and $p_{(k+1)j} \rightarrow p'_{ij}$ by the edges $p_{ik} \rightarrow p_{ij}$, $p_{(k+1)j} \rightarrow p_{ij}$, and $p_{ij} \rightarrow p'_{ij}$. Since $p'_{ij}$ must be closer to $r$ than $p_{ij}$, it follows from the definition of $D$ in Equation~\ref{eqn:spiraldist} that this operation shortens the tree. Contradiction.
\end{proof}
The length of the optimal spiral tree is not necessarily given by $F(1, n)$, but it can also be any of the lengths $F(i, i-1)$ for $2 \leq i \leq n$, so we need to compute the minimum of all these values. Note that there must be at least one wedge $w_{i(i-1)}$ that contains the entire optimal spiral tree, so this will give the length of the optimal spiral tree. Using additional information we can also compute the optimal spiral tree itself in this way. From the definition of $F(i, j)$, it is clear that the algorithm runs in $O(n^3)$ time.

\subsection{Approximation algorithm} \label{sec:approximation}

As shown in Section~\ref{sec:npproof}, computing the optimal spiral tree is NP-hard in general. In this section we describe a simple algorithm that computes a $2$-approximation of the optimal spiral tree. Note that, using Theorem~\ref{thm:fluxapprox}, this algorithm also directly computes a $(2 \sec(\alpha))$-approximation of the optimal flux tree.

For rectilinear Steiner arborescences, Rao \etal~\cite{Rao92} describe a simple $2$-approximation algorithm. The transformation mentioned in Section~\ref{sec:transformation} does not preserve length, so we cannot use this algorithm for spiral trees. However, below we show how to use the same global approach---sweep over the terminals from the outside in---to compute a $2$-approximation for optimal spiral trees in $O(n \log n)$ time.

The basic idea is to iteratively join two nodes, possibly using a Steiner node, until all terminals are connected in a single tree $T$, the \emph{greedy spiral tree}. Initially,
$T$ is a forest. We say that a node (or terminal) is \emph{active} if it does not have a parent in $T$. In every step, we join the two active nodes for which the \emph{join point}
is farthest from $r$. The join point $p_{uv}$ of two nodes $u$ and $v$ is the farthest point $p$ from $r$ such that $p \in \mathcal{R}_u \cap \mathcal{R}_v$. This point is unique if $u$, $v$ and $r$ are not collinear.

\clearpage 

\begin{lemma}
\label{lem:greedyplanar} The greedy spiral tree is planar.
\end{lemma}

\begin{wrapfigure}[6]{r}{.2\textwidth}
  \centering
  \includegraphics{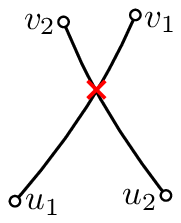}
\end{wrapfigure}
\noindent{\bf Proof.} Assume there is a crossing in the greedy spiral tree between two spiral segments, one between $u_1$ and its parent $v_1$, and another between $u_2$ and its
parent $v_2$. Note that the intersection must be farther from $r$ than both $v_1$ and $v_2$. But that means that the intersection must have been
encountered while both $u_1$ and $u_2$ were in $\mathcal{W}$, so this intersection should be a node in the greedy spiral tree. Contradiction.\hfill\QED
\medskip

\begin{wrapfigure}[9]{r}{.45\textwidth}
  \centering
  \includegraphics[width=.4\textwidth]{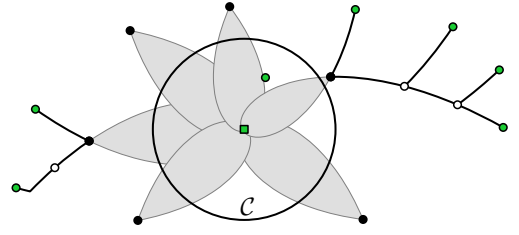}
  \small{\caption{The wavefront $\mathcal{W}$.\label{fig:Wavefront}}}
\end{wrapfigure}
The algorithm sweeps a circle $\mathcal{C}$, centered at $r$, inwards over all terminals. All active nodes that lie outside of $\mathcal{C}$ form the \emph{wavefront}
$\mathcal{W}$ (the black nodes in Figure~\ref{fig:Wavefront}). $\mathcal{W}$ is implemented as a balanced binary search tree, where nodes are sorted according to the radial order
around $r$. We join two active nodes $u$ and $v$ as soon as $\mathcal{C}$ passes over $p_{uv}$. For any two nodes $u, v \in \mathcal{W}$ it holds that $u \notin \mathcal{R}_v$.
By Lemma~\ref{lem:greedyplanar} the greedy spiral tree is planar, so we can apply Lemma~\ref{lem:DPleaforder} to the nodes in $\mathcal{W}$. Hence, when $\mathcal{C}$ passes over $p_{uv}$ and both nodes $u$
and $v$ are still active, then $u$ and $v$ must be neighbors in $\mathcal{W}$. We process the following events.
\begin{description}\itemsep0pt
  \item[Terminal.] When $\mathcal{C}$ reaches a terminal $t$, we add $t$ to $\mathcal{W}$. We need to check whether there exists a neighbor $v$ of $t$ in $\mathcal{W}$ such that $t \in \mathcal{R}_v$. If such a node $v$ exists, then we remove $v$ from $\mathcal{W}$ and connect $v$ to $t$. Finally we compute new join point events for $t$ and its neighbors in $\mathcal{W}$.
  \item[Join point.] When $\mathcal{C}$ reaches a join point $p_{uv}$ (and $u$ and $v$ are still active), we connect $u$ and $v$ to $p_{uv}$. Next, we remove $u$ and $v$ from $\mathcal{W}$ and we add $p_{uv}$ to $\mathcal{W}$ as a Steiner node. Finally we compute new join point events for $p_{uv}$ and its neighbors in $\mathcal{W}$.
\end{description}
We store the events in a priority queue $\mathcal{Q}$, ordered by decreasing distance to $r$. Initially $\mathcal{Q}$ contains all terminal events. Every join point event adds a
node to $T$ and every node generates at most two join point events, so the total number of events is $O(n)$. We can handle a single event in $O(\log n)$ time, so the total running
time is $O(n \log n)$.
Next we prove that the greedy spiral tree is an approximation of the optimal spiral tree.
%


\begin{lemma}
\label{lem:greedycircleisects} Let $\mathcal{C}$ be any circle centered at $r$ and let $T$ and $T'$ be the optimal spiral tree and the greedy spiral tree, respectively. Then $|\mathcal{C} \cap T'| \leq 2 |\mathcal{C} \cap T|$ holds where $|\mathcal{C} \cap T'|$ is the number of intersection points between $\mathcal{C}$ and $T'$.
\end{lemma}
\begin{proof}
It is easy to see that $|\mathcal{C} \cap T'| = |\mathcal{W}|$ when the sweeping circle is $\mathcal{C}$. Let the nodes of $\mathcal{W}$ be $u_1,
\ldots, u_k$, in radial order. Any node $u_i$ is either a terminal or it is the intersection of two spirals originating from two terminals, which we call $u^L_i$ and $u^R_i$ (see
Figure~\ref{fig:GreedyProof}). We can assume the latter is always the case, as we can set $u^L_i = u_i = u^R_i$ if $u_i$ is a terminal. Next, let the intersections of $T$ with
$\mathcal{C}$ be $v_1, \ldots, v_h$, in the same radial order as $u_1, \ldots, u_k$. As $T$ has the same terminals as $T'$, every terminal $u^L_i$ and $u^R_i$ must be able to reach a point $v_j$. Let
$I^L_i$ and $I^R_i$ be the reachable parts (intervals) of $\mathcal{C}$ for $u^L_i$ and $u^R_i$, respectively (that is $I^L_i = \mathcal{C} \cap \mathcal{R}_{u^L_i}$ and $I^R_i = \mathcal{C} \cap \mathcal{R}_{u^R_i}$). Since any two neighboring nodes $u_i$ and $u_{i+1}$ have not been joined by the greedy algorithm, we know that $I^L_i \cap I^R_{i+1} = \emptyset$. Now consider the
collection $\mathcal{S}_j$ of intervals that contain $v_j$. We always treat $I^L_i$ and $I^R_i$ as different intervals, even if they coincide. The union of all $\mathcal{S}_j$ has cardinality $2 k$. If $|\mathcal{S}_j| \geq 5$, then its intervals cannot be
consecutive (i.e. $I^L_i, I^R_i, I^L_{i+1}, I^R_{i+1}$, etc.), as this would mean it contains both $I^L_i$ and $I^R_{i+1}$ for some $i$. So say the intervals of $\mathcal{S}_j$ are not consecutive and $\mathcal{S}_j$ contains $I^L_{i}$
and $I^L_{i+1}$, but not $I^R_i$ (other cases are similar). $T'$ is planar, so this is possible only if $I^R_i \subset I^L_{i+1}$ (see Figure~\ref{fig:GreedyProof2}). But then
$I^R_i$ and $I^L_{i+1}$ are both in a collection $\mathcal{S}_{j'}$ and we can remove $I^L_{i+1}$ from $S_j$, while keeping the union of all collections the same. We repeat this process to construct reduced collections $\hat{\mathcal{S}}_j$ such
that the union of all collections remains the same and all intervals in a collection $\hat{\mathcal{S}}_j$ are consecutive. As a result, $|\hat{\mathcal{S}}_j| \leq 4$, and hence $4h \geq 2k$
or $k \leq 2h$.
\end{proof}

\begin{figure}[t]
  \centering
  \begin{minipage}[t]{.6\textwidth}
  \centering
  \includegraphics[scale=0.85]{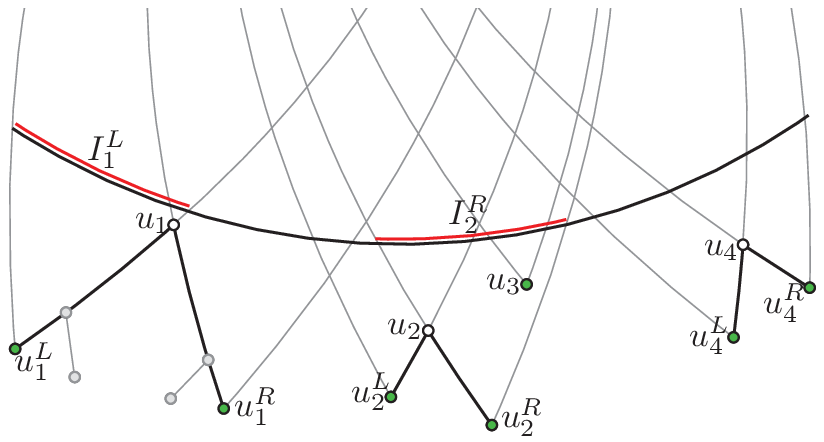}
  \caption{Nodes $u_i \in \mathcal{W}$, terminals $u^L_i, u^R_i$ and intervals $I^L_i, I^R_i$.}
  \label{fig:GreedyProof}
  \end{minipage}
  \hfill
  \begin{minipage}[t]{.35\textwidth}
  \centering
  \includegraphics[scale=0.85]{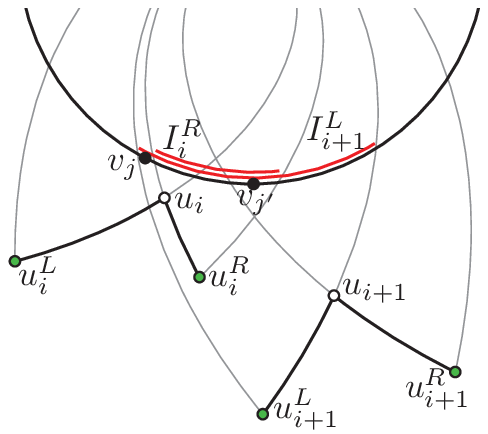}
  \caption{$I^R_i \subset I^L_{i+1}$.}
  \label{fig:GreedyProof2}
  \end{minipage}
\end{figure}

\begin{theorem}
The greedy spiral tree is a $2$-approximation of the optimal spiral tree and can be computed in $O(n \log n)$ time.
\end{theorem}
\begin{proof}
The time bound is already mentioned above. For the approximation, recall that $L(T) = \sec(\alpha) \int_0^\infty |T \cap \mathcal{C}_R| dR$, where $T$ is any spiral tree and
$\mathcal{C}_R$ is the circle of radius $R$ centered at $r$. Using Lemma~\ref{lem:greedycircleisects}, we can directly conclude that the greedy spiral tree is a $2$-approximation
of the optimal spiral tree.
\end{proof}

\noindent The approximation factor is most likely not tight. Experiments for rectilinear Steiner arborescences show that the greedy algorithm
often computes near-optimal arborescences \cite{Cordova94}.

\section{Approximating spiral trees in the presence of obstacles}\label{sec:obstacles}

In this section we extend the approximation algorithm of Section~\ref{sec:approximation} to include obstacles. Given the similarities between spiral trees and rectilinear Steiner arborescences described in Section~\ref{sec:transformation}, it makes sense to consider existing algorithms for rectilinear Steiner arborescences in the presence of obstacles. Unfortunately, the only known algorithm for this seems to have some issues. We discuss these issues in the next section. Then we give a new algorithm for computing rectilinear Steiner arborescences in the presence of obstacles. For a certain type of obstacles, this algorithm also computes a $2$-approximation of the optimal rectilinear Steiner arborescence, although this does not hold for general obstacles. Finally we extend this algorithm to compute spiral trees in the presence of obstacles, again computing a $2$-approximation for a certain type of obstacles.

\subsection{Ramnath's algorithm}\label{sec:ramnath}

\begin{wrapfigure}[7]{r}{.15\textwidth}
  \centering
  \includegraphics[scale=.8]{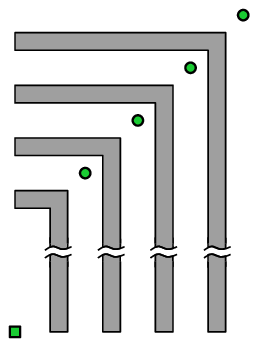}
\end{wrapfigure}
Ramnath~\cite{Ramnath03} gives a $2$-approximation algorithm for rectilinear Steiner arborescences with rectangular obstacles. He claims that the result extends to arbitrary rectilinear obstacles. But this is not the case. Consider the configuration of points and obstacles as seen on the right. The obstacles are $L$-shaped with the longer (vertical) side of the $L$ being much longer than the shorter (horizontal) one. Between each consecutive pair of obstacles their is a terminal. What Ramnath's algorithm does is to sweep a line of slope $-1$ starting at the root. During the sweep the arborescence is constructed greedily maintaining a minimal set of points (called \emph{cover points}) on the sweep line such that all remaining points can still be connected. Thus, in the beginning the algorithm has to decide whether to grow the arborescence to the right or upwards. From these two options the algorithm picks an arbitrary one, in particular it might grow to the right. But then on the arborescence will connect to each terminal by a connection corresponding to the longer side of the $L$-shape. By making the $L$-shape sufficiently long, the approximation factor for this configuration can be made worse than any constant, in particular two.

Ramnath's paper also lacks the details to establish the claimed running time for rectangular obstacles. In particular the subdivision of a critical region (that is, a region that can be exclusively reached by one of the cover points) seems to assume that there is no obstacle strictly inside the critical region. However this case might occur and it does not seem straightforward to extend the algorithm to handle this case. Furthermore, the algorithm needs to compute the point at which the critical regions of neighboring cover points meet. This point is found by tracing paths from both of the cover points. The cost of this tracing step does not seem to be handled in the analysis and it is not clear how to account for it.

\subsection{Rectilinear Steiner arborescences}\label{sec:rectisweep}

We are now given a root $r$ at the origin, terminals $t_1, \ldots, t_n$ in the upper-right quadrant, and also $m$ polygonal obstacles $B_1, \ldots, B_m$ with total complexity $M$. We place a bounding square around all terminals and the root and consider the ``free space'' between the obstacles as a polygonal domain $P$ with $m$ holes and $M+4$ vertices. We describe a greedy algorithm that computes a rectilinear Steiner arborescence $T$, the \emph{greedy arborescence}, inside $P$. Our algorithm returns only a topological representation of $T$. This can easily be extended to the explicit arborescence, which, however, can have arbitrarily high complexity.

As before we incrementally join nodes until we have a complete arborescence. This time we sweep a diagonal line $L$ over $P$ towards $r$ and maintain a wavefront $\mathcal{W}$ with all active nodes that $L$ has passed. If $L$ reaches a join point $p_{uv}$ of nodes $u, v \in \mathcal{W}$, we connect $u$ and $v$ to $p_{uv}$ and add the new Steiner node to $\mathcal{W}$. Our greedy arborescence is restricted to grow inside the polygonal domain $P$. If a point $p \in P$ cannot reach $r$ with a monotone path in $P$, then $p$ is not a suitable join point. To simplify matters we compute a new polygonal domain $P'$ from $P$, such that for every $p \in P'$, there is a monotone path from $p$ to $r$ in $P$. For now we simply assume that we are given $P'$ and that it has $O(M)$ vertices.

\begin{wrapfigure}[6]{r}{.2\textwidth}
  \centering
  \includegraphics{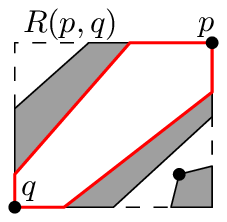}
\end{wrapfigure}
To compute join points we keep track of the reachable region of every node $u \in \mathcal{W}$, that is, we keep track of the part of $L$ that can be reached from $u$ via a monotone path in $P'$. As soon as two nodes $u, v \in \mathcal{W}$ can reach the same point $p$ on $L$, then $p$ is the join point $p_{uv}$ and we can connect $u$ and $v$ to $p_{uv}$. To compute the path between $u$ and $v$ and $p_{uv}$, we need some additional information. Here our definitions follow Mitchell~\cite{Mitchell92}. Given two points $p, q \in P'$ (with $x_q \leq x_p$
and $y_q \leq y_p$), let $R(p, q)$ be the rectangle with $p$ and $q$ as corners. We say that $q$ is \emph{immediately accessible} from $p$ if $p$
and $q$ are in the same connected component of $R(p, q) \cap P'$ and this connected component does not
contain any other vertices or nodes. 
The \emph{parent} of a point $p \in P'$
is the rightmost vertex or node from which $p$ is immediately accessible.
The topological representation of the greedy arborescence stores only the parent information.


The status of the sweep line $L$ consists of three types of intervals: (i) \emph{free intervals}: points that cannot be reached by any node in
$\mathcal{W}$, (ii) \emph{obstacle intervals}: points not in $P'$, and (iii) \emph{reachable intervals}: points reachable by a node in $\mathcal{W}$. The latter type of interval is tagged with the unique node in $\mathcal{W}$ that can reach this interval. We split the reachable intervals such that every interval has a
unique parent. The intervals are stored by their endpoints in a balanced binary search tree. Initially, the status of $L$ consists of one obstacle interval. We distinguish three types of events, which are processed in order using a priority queue.

\noindent{\bf Terminal event.} When we encounter a terminal $t_i$, there are two cases. Either the terminal is in a free interval or in a reachable interval tagged by a node $u$. In the latter case, we connect $u$ to $t_i$ (using the parent information) and replace $u$ by $t_i$ in $\mathcal{W}$. Also, we replace all intervals tagged with $u$ by free intervals and merge them where possible. In both cases, we start a new interval for $t_i$. For the endpoints of this interval, we trace the intersections between $L$ and the horizontal and vertical line through $t_i$. Note that we also split an interval, so we add three intervals in total and remove one. For every new interval (or merged interval), we add vanishing events to the event queue.

\begin{wrapfigure}[8]{r}{.25\textwidth}
\centering
\includegraphics{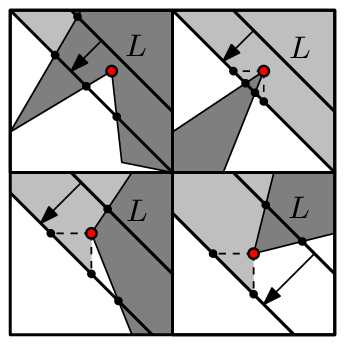}
\end{wrapfigure}
\noindent{\bf Vertex event.} When we encounter a vertex $v$, then $v$ can be in any type of interval. If $v$ is in a free interval, then we add an obstacle interval, where the endpoints of the interval trace the edges of $P'$ connected to $v$. If $v$ is in an obstacle interval, then we add a free interval, where the endpoints of the interval trace the edges of $P'$ connected to $v$. Otherwise, $v$ is in a reachable interval or at the endpoint between a reachable interval and an obstacle interval. In the first case, we need to insert an obstacle interval at $v$, as described above. In both cases we need to set the parent of $v$ and insert a new reachable interval for $v$ (with the correct tag). Also, we need to follow the edge or edges of $P'$ connected to $v$. This can create free intervals. If one of the endpoints of the reachable interval of $v$ directly moves out of $P'$, we do not need to add this endpoint, but we can use the endpoint of the obstacle interval instead. Note that we add only a constant number of intervals. For the new intervals, we add vanishing events to the event queue.

\noindent{\bf Vanishing Interval.} If an interval $I$ vanishes, then there are different cases depending on the types of the neighboring intervals $I_1$ and $I_2$. Note that $I$ vanishes at a point $p$ where two endpoints meet. If $I_1$ and $I_2$ are reachable intervals with different tags $u_1$ and $u_2$, then $p$ is the join point for $u_1$ and $u_2$. We join $u_1$ and $u_2$ at $p$, as described in the terminal event. Otherwise, we need to remove one of the two endpoints. An endpoint of an interval always follows an edge of $P'$ or a vertical or horizontal line through a node in $\mathcal{W}$ or a vertex of $P'$. If $I_1$ and $I_2$ are obstacle intervals or free intervals, then we can just remove both endpoints of $I$. If $I_1$ and $I_2$ are reachable intervals with the same tag, then we keep the endpoint that follows a horizontal line (this follows the definition of a parent given above).
If $I_1$ and $I_2$ are of different types, then we keep the endpoint of the obstacle interval if one is present and otherwise we keep the endpoint of the reachable interval. Again, we add vanishing event points to the event queue for every interval for which an endpoint has changed.

\smallskip
\noindent
The algorithm terminates when $L$ reaches $r$, at which point we have one node left in $\mathcal{W}$. Using the parent information in the status, we
connect the final node with $r$. To compute $P'$ from $P$ we simply run the sweep line algorithm in the opposite direction, tracing the ``reachable region'' of $r$. The points that border a reachable interval and either a free or obstacle interval trace out $P'$.

\begin{lemma}
\label{lem:arbruntime} The greedy arborescence can be computed in $O((n+M) \log(n+M))$ time.
\end{lemma}
\begin{proof}
First we give a bound for the number of events. Clearly, the number of terminal and vertex events are bounded by $O(n+M)$. This also means that the total number of intervals is
bounded by $O(n+M)$, as we add a constant number of intervals at only these events. At every vanishing interval event we remove an interval, so the total number of events is
$O(n+M)$. Also note that every event can generate only a constant number of events. This also means that $P'$ has complexity $O(M)$, as we add vertices to $P'$ only at events. It is
easy to see that all events can be executed in $O(\log(n+M))$ time, except when we need to change all intervals tagged by a certain node $u$ to free intervals. We can do this in
$O(n_u)$ time (by simple bookkeeping), where $n_u$ is the number of intervals tagged by $u$. An interval can only once be changed to a free interval. Merging two neighboring free
intervals removes one interval, so we can charge these operations to the total number of intervals. Furthermore, the topological representation of the greedy arborescence contains only the relevant vertices and nodes to compute the paths between nodes. Every vertex or node can occur only once in this representation. So the algorithm runs in $O((n+M) \log(n+M))$ time.
\end{proof}

If $P$ has only \emph{positive monotone} holes, then the greedy arborescence is a $2$-approximation of the optimal rectilinear Steiner arborescence. A hole is
\emph{positive monotone} if its boundary contains two points $p$ and $q$ such that both paths on the boundary from $p$ to $q$ are monotone in both the $x$-direction and the $y$-direction. In the next section we prove this result for spiral trees. The same arguments can directly be applied to prove the same result for rectilinear Steiner arborescences.

\begin{theorem}
\label{lem:MSA}
The greedy arborescence can be computed in $O((n+M) \log(n+M))$ time. If $P$ has only positive monotone holes, then the greedy arborescence is a $2$-approximation of the optimal rectilinear Steiner arborescence.
\end{theorem}

\subsection{Spiral trees} \label{sec:spirobstacles}

We now describe how to adapt our algorithm to spiral trees; we concentrate mainly on the necessary changes.
We again compute only a topological representation of the output and refer to the spiral tree which we compute as the \emph{greedy spiral tree}. The sweep line is replaced by a sweeping circle $\mathcal{C}$. A simple balanced binary search tree is still sufficient to store the intervals, using special cases to deal with the circular topology.

\begin{wrapfigure}[6]{r}{.27\textwidth}
  \centering
  \includegraphics{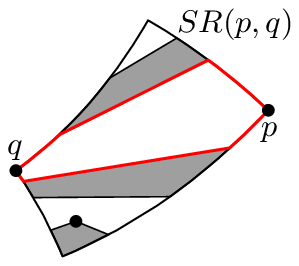}
\end{wrapfigure}
We need to replace horizontal and vertical lines by right and left spirals. For a given node or vertex $u$, the endpoints of its interval on $\mathcal{C}$ follow the intersections
of $\mathcal{S}^{+}_u$ and $\mathcal{S}^{-}_u$ with $\mathcal{C}$. Given two points $p$ and $q$ ($q \in \mathcal{R}_p$), let $SR(p,q)$ be the \emph{spiral rectangle} between $p$
and $q$. The spiral rectangle between $p$ and $q$ is bounded by the two paths (these are unique) consisting of exactly two spiral segments connecting $p$ to $q$ (this is exactly a
rectangle transformed by the transformation in Section~\ref{sec:transformation}). The point $q$ is \emph{immediately accessible} from $p$ if $p$ and $q$ are in the same connected
component of $SR(p, q) \cap P'$ and this connected component does not contain any other vertices or nodes.

\begin{wrapfigure}[8]{r}{.27\textwidth}
  \centering
  \includegraphics{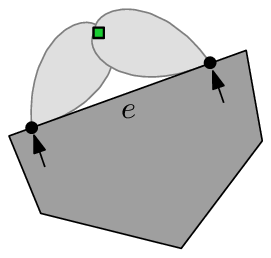}
  \small{\caption{Spiral points.\label{fig:SubdividePoly}}}
\end{wrapfigure}

There is one subtlety. If the left or right spiral of a vertex $v$ directly moves out of $P$, we can ignore it, as before. However, at the exact moment that this is no longer the case, we do need to trace this spiral. For rectilinear Steiner arborescences, this can only happen at vertices. For spiral trees, this can also happen at most two \emph{spiral points} in the middle of an edge $e$. A point $p$ on $e$ is a spiral point if the angle between the line from $p$ to $r$ and the line through $e$ is exactly $\alpha$. We hence subdivide every edge of $P$ at the spiral points. In addition we also subdivide $e$ at the closest point to $r$ on $e$ to ensure that every edge of $P$ has a single intersection with $C$.

Neither the algorithm presented in Section~\ref{sec:rectisweep} nor its adaptation to spiral trees gives a constant factor approximation. But, if we restrict the types of obstacles, they give 2-approximations. For rectilinear Steiner arborescences we have to use positive monotone obstacles, for spiral trees \emph{spiral monotone} obstacles. An obstacle is \emph{spiral monotone} if its boundary contains two points $p$ and $q$ such that both paths on the boundary from $p$ to $q$ are angle-restricted.


\begin{lemma}\label{lem:spiralmonotone}
Let $P$ be a polygonal domain with spiral monotone holes. Then all points on a circle $\mathcal{C}$ reachable from a node $u$ lie inside a single circular interval $I_u \subseteq \mathcal{C}$ with the property that every point in $I_u \cap P$ is reachable from $u$.
\end{lemma}
\begin{proof}
Consider a point $p \in I_u \cap P$. Let $\pi_1$ and $\pi_2$ be the paths from $u$ to the endpoints of $I_u$. Repeat the following until we hit either $\pi_1$ or $\pi_2$. Move along the left spiral through $p$ going outwards (from $r$). When we hit a hole, simply follow the outline of the hole until we can follow the left spiral again. Because $P$ has only spiral monotone holes, we eventually reach either $\pi_1$ or $\pi_2$. Hence $p$ must be reachable from $u$.
\end{proof}

\begin{theorem}
\label{theo:greedy}
The greedy spiral tree can be computed in $O((n+M) \log(n+M))$ time. If $P$ has only spiral monotone holes, then the greedy spiral tree is a $2$-approximation of the optimal spiral tree.
\end{theorem}
\begin{proof}
Correctness and running time follow from Lemma~\ref{lem:arbruntime} and the discussion in Section~\ref{sec:spirobstacles}.
Assume that $P$ has only spiral monotone holes and let $T$ and $T'$ be the optimal and greedy spiral tree, respectively. By Lemma~\ref{lem:spiralmonotone} we can represent the part of a circle $\mathcal{C}$ that is reachable by a terminal $t$ as a single interval $I_t$. We can now follow the proof of Lemma~\ref{lem:greedycircleisects} with these intervals to show that $|\mathcal{C} \cap T'| \leq 2 |\mathcal{C} \cap T|$. This directly implies that, if $P$ has only spiral monotone holes, the greedy spiral tree is a $2$-approximation.
\end{proof}



\begin{thebibliography}{21}

\bibitem{FlowMapper}
{CSISS - Spatial Tools: Tobler's Flow Mapper}.
\newblock \url{http://www.csiss.org/clearinghouse/FlowMapper}.

\bibitem{Aichholzer2001}
O.~Aichholzer, F.~Aurenhammer, C.~Icking, R.~Klein, E.~Langetepe, and G.~Rote.
\newblock Generalized self-approaching curves.
\newblock {\em Discr. Appl. Mathem.}, 109(1-2):3--24, 2001.

\bibitem{Awerbuch1990}
B.~Awerbuch, A.~Baratz, and D.~Peleg.
\newblock Cost-sensitive analysis of communication protocols.
\newblock In {\em Proc. 9th ACM Symposium on Principles of Distributed
  Computing}, pages 177--187. ACM, 1990.

\bibitem{Boyandin2010}
I.~Boyandin, E.~Bertini, and D.~Lalanne.
\newblock Using Flow Maps to Explore Migrations Over Time.
\newblock In {\em Proc. Workshop in Geospatial Visual Analytics: Focus on Time (GeoVA(t))}, 2010.

\bibitem{Brazil2001}
M.~Brazil, J.~H. Rubinstein, D.~A. Thomas, J.~F. Weng, and N.~C. Wormald.
\newblock Gradient-constrained minimum networks. {I. Fundamentals}.
\newblock {\em Journal of Global Optimization}, 21:139--155, October 2001.

\bibitem{Brazil2007}
M.~Brazil and D.~A. Thomas.
\newblock Network optimization for the design of underground mines.
\newblock {\em Networks}, 49:40--50, January 2007.

\bibitem{InfoVisFlowMap}
K.~Buchin, B.~Speckmann, and K.~Verbeek.
\newblock Flow map layout via spiral trees.
\newblock Accepted InfoVis 2011,
  \url{http://www.win.tue.nl/~speckman/InfoVisFlowMaps.pdf}.

\bibitem{Cordova94}
J.~C\'{o}rdova and Y.~Lee.
\newblock A heuristic algorithm for the rectilinear {Steiner} arborescence
  problem.
\newblock Technical report, Engineering Optimization, 1994.

\bibitem{Dent1999}
B.~D. Dent.
\newblock {\em Cartography: Thematic Map Design}.
\newblock McGraw-Hill, New York, 5th edition, 1999.

\bibitem{Guo2009}
D.~Guo.
\newblock Flow Mapping and Multivariate Visualization of Large Spatial Interaction Data.
\newblock {\em IEEE Transactions on Visualization and Computer Graphics}, 15(6):1041--1048, 2009.

\bibitem{Krozel2006}
J.~Krozel, C.~Lee, and J.~Mitchell.
\newblock Turn-constrained route planning for avoiding hazardous weather.
\newblock {\em Air Traffic Control Quarterly}, 14(2):159--182, 2006.

\bibitem{Lu2000}
B.~Lu and L.~Ruan.
\newblock Polynomial time approximation scheme for the rectilinear {Steiner}
  arborescence problem.
\newblock {\em J. Comb. Optimization}, 4(3):357--363, 2000.

\bibitem{Mitchell92}
J.~Mitchell.
\newblock {$L_1$} shortest paths among polygonal obstacles in the plane.
\newblock {\em Algorithmica}, 8:55--88, 1992.

\bibitem{Phan2005}
D.~Phan, L.~Xiao, R.~Yeh, P.~Hanrahan, and T.~Winograd.
\newblock Flow map layout.
\newblock In {\em Proc. IEEE Symposium on Information Visualization}, pages
  219--224, 2005.

\bibitem{Ramnath03}
S.~Ramnath.
\newblock New approximations for the rectilinear {Steiner} arborescence
  problem.
\newblock {\em IEEE Trans. Computer-Aided Design Integ.
  Circuits Sys.}, 22(7):859--869, 2003.

\bibitem{Rao92}
S.~Rao, P.~Sadayappan, F.~Hwang, and P.~Shor.
\newblock The rectilinear {Steiner} arborescence problem.
\newblock {\em Algorithmica}, 7:277--288, 1992.

\bibitem{Shi2000}
W.~Shi and C.~Su.
\newblock The rectilinear {Steiner} arborescence problem is {NP}-complete.
\newblock In {\em Proc. 11th ACM-SIAM Symposium on Discrete Algorithms}, pages
  780--787, 2000.

\bibitem{ss-rsap-05}
W.~Shi and C.~Su.
\newblock The rectilinear {Steiner} arborescence problem is {NP}-complete.
\newblock {\em SIAM Journal on Computing}, 35(3):729--740, 2005.

\bibitem{Slocum2010}
T.~A. Slocum, R.~B. McMaster, F.~C. Kessler, and H.~H. Howard.
\newblock {\em Thematic Cartography and Geovisualization}.
\newblock Pearson, New Jersey, 3rd edition, 2010.

\bibitem{WaldoTobler1987}
W.~Tobler.
\newblock Experiments in migration mapping by computer.
\newblock {\em The American Cartographer}, 14(2):155--163, 1987.

\bibitem{Wood2010}
J.~Wood, J.~Dykes, and A.~Slingsby.
\newblock Visualization of Origins, Destinations and Flows with {OD} Maps.
\newblock {\em The Cartographic Journal}, 47(2):117--129, 2010.

\end{thebibliography}

\end{document}